\documentclass[sigconf,letter,preprint]{acmart}

\renewcommand\footnotetextcopyrightpermission[1]{} 
\makeatletter
\renewcommand\@formatdoi[1]{\ignorespaces}
\makeatother

\setcopyright{none}

\acmConference[]{}{}{}
\acmYear{}
\copyrightyear{}

\usepackage{amsfonts,amsmath,amsthm,amssymb,dsfont}

\usepackage{caption}
 \usepackage{xcolor}

\newcommand{\E}{\mathbf{E}}

\newcommand{\R}{\mathbb{R}}

\newcommand{\SF}{\mathrm{SF}}

\newtheorem{remark}{Remark}

\makeatletter
\def\@copyrightspace{\relax}
\makeatother

\begin{document}
\title{Analyzing LoRa
long-range, 
low-power, wide-area networks 
 using stochastic geometry}

\author{Bart{\l}omiej B{\l}aszczyszyn}
\affiliation{%
  \institution{Inria/ENS}
  \streetaddress{Rue Simone Iff}
  \city{Paris}
  \postcode{75589}
}
\email{Bartek.Blaszczyszyn@ens.fr}

\author{Paul M\"uhlethaler}
\affiliation{%
  \institution{Inria}
 \streetaddress{Rue Simone Iff}
  \city{Paris}
  \postcode{75589}
}
\email{Paul.Muhlethaler@inria.fr}

\begin{abstract}
In this paper we present a simple, stochastic-geometric model of 
a wireless access network exploiting the LoRA (Long Range) protocol, which is a
non-expensive technology allowing for  long-range, single-hop connectivity 
for the Internet of Things. We assume a space-time Poisson model of packets
transmitted by LoRA  nodes to a fixed base station. 
Following previous studies of the impact of  interference~\cite{haab17,boro16},  
we assume that a given packet is successfully received when no interfering packet arrives with similar power before the given packet  payload phase. This is as a consequence of LoRa using different 
transmission rates for different link budgets (transmissions with smaller received powers use larger spreading factors) and 
LoRa intra-technology interference treatment.
Using our model, we study the scaling of the  packet reception probabilities  
per  link budget  as a function  of  the  spatial  density  of  nodes  and  
their rate  of  transmissions.  We consider both the parameter values 
recommended by the LoRa provider, as well as proposing LoRa tuning  to improve the equality of  performance for all link budgets. We also consider spatially
non-homogeneous distributions of LoRa nodes. We show also how a fair 
comparison to non-slotted Aloha can be made within the same  framework.
\end{abstract}


\maketitle

\section{Introduction}

Low-power, wide-area networks (LPWANs) will undoubtedly play a crucial role in the development of the 
Internet of Things. They are wireless cellular networks that,  in contrast to  LTE  (and also WiFi and Bluetooth), 
 do not focus on  high data rate communications. The goal of LPWANs is to ensure 
large coverage areas with reasonable data rates and low energy consumption. A good 
scaling in terms of the density of nodes is also a key requirement for these systems. 

Several systems in the field of LPWANs exist. 
Essentially, four options are available to customers:
Sigfox,  operator standards LTE-M or NB-IoT, and  several  proprietary solutions, LoRa (Long Range) being one of them.
Sigfox is 
offered as an operator service and the technology is proprietary. It operates in the 
968/902 MHz licence-free industrial scientific and medical (ISM) radio band. 
Sigfox acts as both the technology and the service provider and has quickly deployed 
a large number of base stations in Europe and beyond. Long Term Evolution Machine-to-Machine
LTE-M and  Narrow Band-IoT are the standards promoted by the traditional telecom operators 
that expect to propose solutions for  machine-to-machine traffic through standardized 
solutions, in particular in the framework of the forthcoming 5G systems. There is a significant number of 
proprietary solutions such as Accelus, Cyan's Cynet, Ingenu/On-Ramp, SilverSpring Starfish, etc.,
but only a limited amount of information about these systems is available. 

LoRa is also a proprietary technology developed by Cycleo of Grenoble, France, and acquired by Semtech in 2012.
As its name suggests, it mainly aims at long range transmissions with a high robustness, multipath 
and Doppler resistance, and low power.  LoRa networks often operate in the unlicensed ISM band (434 MHz and 
868 MHz in Europe and 433 MHz and 915 MHz in the USA), but can also use licensed 
bands in the spectrum ranging from 137 to 1020 MHz. 
LoRa uses  a kind of Chirp Spread Spectrum~(CSS) with Forward Error Correction (FEC), where wideband linear frequency modulated chirp pulses are used to encode information and the spreading factor is directly linked to 
the power decay between  the end device  and the base station.

We decided to focus our study on LoRa systems because, even though it is 
a proprietary technology, a significant amount of information is available on it,  see Section~\ref{sec_assumptions},  
making it possible to model and study LoRa networks. Moreover, LoRa industrial momentum seems to be
significant and commercial success for LoRa is likely. 

LoRa technology has already been studied, both  for the evaluation of pairs of simultaneous transmissions~\cite{haab17} and for the scaling evaluation of a large network~\cite{mipe16},  \cite{boro16}. However, to the authors' best knowledge, there is no 
commonly accepted analytical model of LoRa networks. 
A major difficulty in addressing performances of LoRa  consists in its 
multiparameter nature, with a crucial  rule played by the spreading factors
and, closely related to them,  packet transmission time.
An additional difficulty is the complexity of the signal-to-interference analysis of packet reception in CSS technology.

In this paper we propose a mathematical model allowing one to calculate  packet reception probabilities per link budget, according to the  spatial density of nodes and their rate of transmissions, as well as 
 all  main  LoRa parameters.
 It is based on  a 
Poisson rain (space-time) model of packet transmissions from the nodes to the base station (gateway). This model has already been used in~\cite{blaszczyszyn2010stochastic,blaszczyszyn2015interference} to study non-slotted Aloha system.

The remainder of this paper is organized as follows. In Section~\ref{sec_rel} we present some related studies, which are crucial to our approach.
In Section~\ref{sec_assumptions} we present  LoRa system assumptions and in Section~\ref{sec_net} we present our model with its analysis.
Section~\ref{sec_num}  presents numerical results.
In Section~\ref{s.Aloha} we briefly sketch  how LoRa can be compared to non-slotted Aloha using a Poisson rain model of packet transmissions.
Finally, Section~\ref{sec_conc} concludes the article.


\section{Related studies}
\label{sec_rel} 

A few papers, such as~\cite{pego16}, 
present measurements of existing LoRa systems and study 
the actual performances of the  end devices with respect to their relative distance to the  gateways.
They aim at optimising the parametrization of  LoRa networks.  However, due to the limited number of  end devices 
considered in these studies, it is difficult to gauge  the scaling properties of LoRa  networks, for instance in 
terms of the maximum number of   end devices supported with a given goodput rate. 
Moreover, these  studies do not  permit a fine-tuned analysis of 
collisions in LoRa networks. 

In contrast,~\cite{haab17} scrutinizes interference in LoRa in order to  
establish packet collision rules, which are  reproduced next  in a simulation model 
 allowing the authors to study  the scalability issues of LoRa networks.
 A real platform to test 
capture in LoRa networks is used in~\cite{boro16}.
The conclusions drawn from the tests lead to a collision model close 
to that derived in~\cite{haab17} and to similar scalability evaluation results. 
 We recall the  collision rules established in~\cite{haab17} and \cite{boro16} in 
 Section~\ref{subsec_coll} since we integrate them in our  stochastic-geometric LoRa network model.


\section{LoRa system description and assumptions}
\label{sec_assumptions}

LoRa systems are highly parameterizable. The carrier frequency can be chosen  
from  137 to 1020 MHz by steps of 61 kHz and  
the bandwidth (the width of frequencies in the radio band)
can be set to 125 MHz, 250 MHz or 500 MHz.
LoRa can use transmission powers from -4 dBm to 20 dBm.

A key parameter of LoRa 
is the spreading factor, which is the ratio  between the symbol rate and the chip rate.
Denoted in what follows by SF, and expressed at the logarithmic scale at base~2 (thus the increase by 1 
of the SF multiplies the chip per symbol rate by 2),
the SF of LoRa varies from 6 to 12. With SF = 6 the symbol rate is the highest but the transmission is the least 
robust.  The value of $\SF$  is fixed for individual transmissions depending on the received power, as explained in Section~\ref{subsec_sf}.

The physical layer of LoRa  also incorporates Forward Error 
Correction (FEC) that permits the recovery of bits of
information in the case of corruption by interference. This requires a small overhead of additional encoding of the data in
the transmitted packet thus reducing the Coding Rate (CR).  The available values of CR
are $4/5,5/6,4/7,4/8$ (with smaller values providing more reliability).

\subsection{Spreading factor and received power}
\label{subsec_sf} 

The spreading factor parameter SF, which  is the logarithm at base~2 of the  ratio  between the symbol and the chip rate,
is a key parameter of  LoRa.  Even if in principle it can be  set arbitrarily, it is obvious that its value must be
appropriately  chosen in order to optimize the performance of the network.
A reasonable choice is to make it inversely related to the received power: higher received power allows one to use a
 smaller spreading factor and thus reduce the packet transmission time, as specified in Section~\ref{subsec_duration}.
As an example, Table~\ref{tab1}  
gives the values of SF proposed in~\cite{sx12_ds}
for different thresholds of the received power, called {\em sensitivity}, for  transmissions with  $125$~kHz bandwidth.

\begin{table}[ht]
\begin{center}
\begin{tabular}{|c|c|}
  \hline
   SF        & Sensitivity   in dBm  \\
       &  (for 125 kHz bandwidth system)  \\
  \hline
  6 & -121\\
   \hline
  7 &  -124 \\
   \hline
  8 &  -127 \\
   \hline
  9 &  -130\\
   \hline
  10& -133\\
   \hline
  11& -135\\
   \hline
  12 & -137\\
  \hline 
\end{tabular}
\vspace{2ex}
\caption{Sensitivity (minimum received power thresholds) and corresponding spreading factors}
\label{tab1}
\end{center}
\end{table}
For example, 
when the  signal is received with at least $-121$~dBm of power then the spreading factor is set to its minimal value $\SF=6$.
When the received power is in the interval $[-124,-121)$ then $\SF=7$, etc. 
The smallest acceptable  received power $-137$~dBm triggers $\SF=12$.  

\subsection{Collision model}
\label{subsec_coll} 

A precise (signal to interference) analysis of the reception of  packets in LoRa chirp spread spectrum systems is complicated.
We will use a simplified {\em collision model}, to this end, inspired by the results presented in~\cite{boro16} and~\cite{haab17}. The first assumption we make is 
that {\em two transmissions using different values of SF do not interfere with each other}; this assumption is retained in the collision model 
described in~\cite{boro16}. In other words, the codes with different spreading factors are assumed orthogonal. 

The second assumption concerns collisions caused by  {\em simultaneous  transmissions 
using the same value of SF}. Different cases of this, identified in~\cite{haab17},  
depend on the time overlap of the (two or more) simultaneous  transmissions and their  received powers, called also  Received Signal Strength Indications (RSSIs).
More specifically, ~\cite{haab17} distinguishes five cases of possible overlap of the transmission period of  a given  packet with an interfering 
transmission, presented in Figure~\ref{collision_mod}, combined with  two cases regarding the difference between their RSSI values.
It shows that the given packet is always successfully received  if its RSSI is significantly  greater than that of the interfering packet or,
in the case of similar RSSI values, when the interfering transmission starts relatively late, precisely  during the payload part of the given packet
(case 5 in Figure~\ref{collision_mod}).
It is not successful when an interfering packet with a similar RSSI starts before the payload part  of the given packet 
(depicted as cases 1 to 4 in Figure~\ref{collision_mod}).


\begin{figure}[ht]
\begin{center}
\centerline{\includegraphics[scale=1.4]{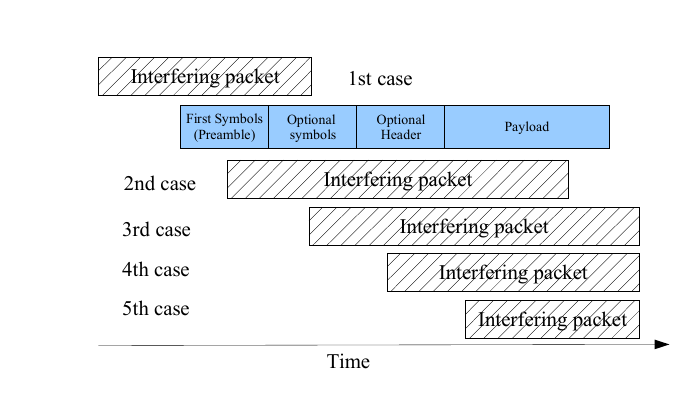}}
\caption{Collision model: different cases of  a given packet transmission (in blue) overlapping with an interfering transmission.
The reception of the given packet is  successful if its received power (RSSI) is significantly greater than that of the interfering packet or,
in the case of similar RSSI values, when the interfering transmission starts during the payload (case 5).
Cases 1--4 with similar RSSI values give collisions.
\label{collision_mod}}
\end{center}
\end{figure}

An additional advantage of relating the spreading factors to the received powers, as described in Section~\ref{subsec_sf},
is that any  two transmissions with significantly different received powers  
will use different  spreading factors and thus, according to the  assumption made above, will not interfere. 
Consequently, we can assume that a {\em given packet is successfully received when no interfering packet 
 arrives with similar power before its payload phase.}

\subsection{Packet transmission times}
\label{subsec_duration} 

The packet transmission (on air) time  depends on its composition (cf Figure~\ref{collision_mod}) and  the symbol transmission time $T_{\mathrm{symbol}}$,
which is related to the spreading factor $\SF$ and the bandwidth $\mathrm{BW}$ via a fundamental relation
\begin{equation}\label{e.Tsymbol}
T_{\mathrm{symbol}}=\frac{2^{\SF}}{\mathrm{BW}}\,.
\end{equation} 
The following specific expressions for transmission times (in seconds) of the two main parts of the packet are given in~\cite{sx12}.

Preamble time (with with the optional symbols): $$(4.25+np)T_{\mathrm{symbol}}.$$
    
    Payload time (with the optional header):
$$\!\!\!\Bigl(8+ \max\Bigl(\Bigl\lceil\frac{8\mathrm{PL}-4\SF+28+16-20\mathrm{H}}{4(\SF-2\mathrm{DE)}}\Bigr\rceil (\mathrm{CR} +4),0 \Bigr)T_{\mathrm{symbol}},$$ 
with notation explained in Table~\ref{tab2}.


\begin{table}[ht]
\begin{center}
\begin{tabular}{|c|c|}
 \hline
   BW   &  bandwidth in Hz  \\
 \hline  
  $\SF$ & Spreading Factor \\
   \hline
  $n_p$    &  number of additional preamble symbols  \\
  \hline
 PL  & payload in Bytes\\
  \hline
  H &  H=0 if header, 1 if no header present  \\
  \hline 
   DE &   DE=1 if low data rate optimization, 0 if not \\
  \hline 
  CR& Coding Rate\\ 
  \hline
\end{tabular}
\vspace{2ex}
\caption{LoRa  main parameters}
\label{tab2}
\end{center}
\end{table}

\section{Network model and analysis}
\label{sec_net}

In this section we propose a stochastic geometric  model of LoRa and present some results of its analysis.

\subsection{Poisson LoRa model with collisions}

Our LoRA network model consists of one base station located (without loss of generality) at the origin of the plane $\R^2$
and a space-time  Poisson process of packets sent to it. 
This model, called {\em Poisson rain} in~\cite{blaszczyszyn2010stochastic,blaszczyszyn2015interference},
assumes packet transmissions  initiated from points $\{X_i\}$ on the plane at time instants $\{T_i\}$, with $\Phi=\{(X_i,T_i)\}$ being
a Poisson point process on $\R^2\times\R$. Our default assumption is that  that $\Phi$ is homogeneous in space and time, 
with 
the expected number of  transmissions initiated per unit of time and space denoted by $\lambda$. 
However, in Section~\ref{subsec_fairness} we shall also consider some space-non-homogeneous, radially symmetric Poisson model of packet  transmissions.

We assume that the packets are transmitted with power $P_{\mathrm{tr}}$ and 
suffer from path-loss modeled by  the usual power-law function $l(r)=(\kappa r)^\beta$, of the transmission distance $r$
with some constants $\kappa>0$ and $\beta>2$, as well as independent propagation effects (such as those from multi-path fading, shadowing or other seemingly random phenomena perturbing the base-station-to-device signal) modeled by a  generi random variable  $F$. 
 Thus, a packet transmitted from a location $X_i$ (regardless of time) is received at the 
base station  with power equal to 
$$P_{\mathrm{rec}}(X_i):= \frac{P_{\mathrm{tr}}F_i}{l(|X_i|)} = 
\frac{P_\mathrm{tr} F_i}{(\kappa |X_i|)^{\beta}}\,$$
where, given $\{X_i\}$, $F_i$ are i.i.d. with common distribution of~$F$ and $|\cdot|$ is the Euclidean distance on the plane.

Following the system description presented in Section~\ref{subsec_sf}, we assume that the   spreading factors $\SF$ used for  transmissions of different packets are related to their received  powers. More specifically,
we consider a set of threshold powers (sensitivities) 
$P_1<P_2<\dots<P_N$, for some $N\ge 1$,
with the corresponding spreading factors $\SF_n$, $n=1,\ldots,n$, and we assume that 
the value $\SF_n$ is used for the transmission from $X_i$
when  $P_\mathrm{rec}(X_i)\in [P_n,P_{n+1})$,
with  $P_{N+1}:=\infty$. No successful reception  (packet loss) is assumed  when $P_\mathrm{rec}(X_i)<P_1$.

We have seen in Ssection~\ref{subsec_duration} that packet transmission times depend on many 
parameters,  in particular the spreading factor $\SF$.
Consequently, we denote by  $B_n$ the transmission time of a packet using spreading factor value $\SF_n$ and by $\Delta_n$ the transmission time of the initial part, required for the base station to  lock  on  the  transmission.
During this latter part,  an interfering packet of similar power can lead to a collision, as explained in~Section~\ref{subsec_coll}.
More precisely, we assume that a given 
packet transmitted from $X_i$
and reaching the base station with power $P_\mathrm{rec}(X_i)\in[P_n,P_{n+1})$
is  successfully received  if no any other packet 
is on air  with received power in   $[P_n,P_{n+1})$ during the initial time  $\Delta_n$ of the given packet's  transmission.

\subsection{Performance analysis}

We begin by recalling a  useful result regarding the 
process of received powers in an homogeneous Poisson network model with power-low path-loss model; see e.g~\cite[Lemma~4.2.2]{blaszczyszyn2018stochastic}.
In the context of our Poisson rain network model it can be formulated as follows. 
\begin{lemma}\label{lemma}
The values of powers received at the origin 
from all packet transmissions initiated in the Poisson rain model $\Phi$ during any time interval of length $T$ form an non-homogeneous Poisson point process on $(0,\infty)$ with intensity measure
$$\frac{2aT}{\beta}t^{-2/\beta-1} dt,$$
where
$$a =\frac{\pi \lambda P^{2/\beta}_{\mathrm{tr}}  E[F^{2/\beta}]}{\kappa^2},$$
$F$ is the generic random propagation in the wireless channel, 
and $\lambda$ is the expected number of  transmissions initiated per unit of time and space.
\end{lemma}

\subsubsection{Successful reception probabilities}
We now consider any given (arbitrary) packet transmission  arriving at  the base station with power within the 
interval, $[P_n,P_{n+1})$ for some $n=1,\ldots,N$.
We denote by $\Pi_n$ the probability  that this transmission is successfully   received by the base station. 
As previously stated, the transmission of this packet takes time $B_n$,
with the initial time $\Delta_n$ required for the base station to lock on this transmission.  
Using Lemma~\ref{lemma} we obtain our first result.
\begin{proposition}\label{prop.Pin}
\begin{align*}
\Pi_n &= \exp \Big(-   a_n (P_n^{-2/\beta}   -  P_{n+1}^{-2/\beta}  )  \Big)\\
\noalign{\text{for $n=1,\ldots,N-1$,  and}}
\Pi_N &= \exp \Big(-   a_N P_N^{-2/\beta} \Big),\\
\noalign{\text{where}} 
a_n& =a (B_n+\Delta_n).
\end{align*}
\end{proposition}
\begin{proof}
According to the adopted collision model, $\Pi_n$ is equal to the probability that there is no transmission in the Poisson rain model during the interval of  time of duration $B_n+ \Delta_n$, reaching the base station with the power
in the same interval  $[P_n,P_{n+1})$,
assuming the convention $P_{N+1}=\infty$. By  Lemma~\ref{lemma}
and the well known formula for Poisson void probabilities 
(cf e.g.,~\cite[Section~3.4.1]{blaszczyszyn2018stochastic})
$$\Pi_n = \exp\Big( \frac{-2a_n}{\beta} \int_{P_n}^{P_{n+1}} t^{-2/\beta-1} dt  \Big).$$
Direct evaluation of the above integral completes the proof.
\end{proof}
Note that the transmission times $B_n$ and $\Delta_n$
depend  on the spreading factor $\SF_n$ used for the transmissions
with received powers in $[P_n,P_{n+1})$,
as well as on other parameters, as explained in Section~\ref{subsec_duration}.

\subsubsection{Equalizing the reception probabilities} 
\label{subsec_fairness} 
We assume now that all LoRa model parameters are given except 
the power thresholds $P_1,\ldots,P_N$
(sensitivities), which we want to fix  
 so that  the packets arriving  with all powers not smaller than~$P_1$ have the same reception probability.  
We have the following result.
\begin{proposition}\label{prop.fairness}
For a given $\Pi\in(0,1)$ assume 
$$P_N= \Big(-\frac{1}{a_N} \log(\Pi)\Big)^{-\beta/2}$$
and 
$$P_{n}= \Big(- \sum_{i=n}^N  \frac{1}{a_i}   \log(\Pi)\Big)^{-\beta/2}$$ 
for $n=1,\ldots,N-1$.
Then the reception probabilities are constant  $\Pi_n=\Pi$.
\end{proposition}
\begin{proof}
The result follows directly from Proposition~\ref{prop.Pin}
by substituting the values of  $P_n$ in the expressions for $\Pi_n$, for $n=1,\ldots,N$.
\end{proof}
Obviously, the larger target value $\Pi$, the larger all power thresholds $P_n$, in particular~$P_1$, meaning that more transmissions will remain unsuccessful because having a received power smaller than the minimal acceptable value~$P_1$.

\subsubsection{A non-homogeneous spatial density of transmission} 
\label{sub_non_hom}
 We consider now a non-homogeneous pattern of transmissions.
 More specifically, let the process of transmissions form a space-time 
 Poisson point process  with the local  density
 of transmission initiated at distance $r$ from the origin being equal to $\lambda r^\alpha$, for some parameter $\alpha$. 
A natural assumption is $\alpha<0$, meaning that the density decays within the distance, but $\alpha>0$ is also possible.
Obviously $\alpha=0$ is the previously considered homogeneous case.

In order to study this case we need the following extension of Lemma~\ref{lemma}.
 \begin{lemma}\label{lemma-inhomo}
The values of powers received during any time interval of length $T$ at the origin in the Poisson rain model, with intensity of transmissions $\lambda_sr^\alpha$ at distance $r$ from the origin,  form an non-homogeneous Poisson point process on $(0,\infty)$ with intensity measure
 $$\frac{2a'T}{\beta}t^{-(\alpha+2)/\beta-1} dt,$$
 where
 $$a =\frac{\pi \lambda P^{(\alpha+2)/\beta}_{\mathrm{tr}}
 \E[F^{(\alpha+2)/\beta}]}{\kappa^{\alpha+2}}.$$  

\end{lemma}
\begin{remark}Note that the intensity of received powers in our 
non-homogeneous model, given in  Lemma~\ref{lemma-inhomo}, coincides 
with that of the homogeneous model presented in Lemma~\ref{lemma}, with 
the path-loss exponent $\beta$ and the 
intensity of transmissions $\lambda$ replaced respectively  by
\begin{align*}
\beta' &= \frac{2\beta}{\alpha+2}\\
\lambda'& = \frac{2\lambda}{(\alpha+2)\kappa^{\alpha}}\,.
\end{align*}
Consequently, we observe that the performance  of the considered  non-homogeneous LoRa network 
(regarding metrics based   on the process of received powers,
as e.g. the reception probabilities $\Pi_n$) is equivalent to  the performance the homogeneous LoRa model
with the above parameter modifications.
This is yet another instance of the network equivalence property formulated in~\cite{netequivalence}.
\end{remark}
\begin{proof}[Proof of Lemma~\ref{lemma-inhomo}]
By the displacement theorem for Poisson processes (see e.g. \cite[Theorem 4.4.2]{blaszczyszyn2018stochastic})
the process of received powers is a Poisson process. 
To  compute its mean measure denoted by $\Lambda(t)$, we calculate 
the mean number of transmissions started   within a unit interval of time and received 
with  power greater than $t>0$ 
$$\Lambda(t) =\E[\#\{(X_i,T_i) \in   \Phi:(\kappa |X_i|)^{-\beta}P_{\mathrm{tr}}F_i > t, T_i\in[0,1]\}].$$
Writing the condition on the received power in the form 
$$ |X_i| < \Big(\frac{P_{\mathrm{tr}}F_i}{t}\Big)^{1/\beta} 1/\kappa $$
and using the spatial density of transmissions, 
we obtain 
\begin{align*}
\Lambda(t) &=   2 \pi \lambda     \E \Big [     \int_0 ^{ (P_{\mathrm{tr}}F_i/t)^{1/\beta}1/\kappa}  r^{\alpha}r dr       \Big]\\
&=   \frac{2 \pi \lambda}{(\alpha+2) \kappa^{\alpha+2}}     \E \Big [ F^{(\alpha+2)/\beta}   \Big]  {t}^{-(\alpha+2)/\beta}.
\end{align*}
Differentiating $\Lambda(t)$ with respect to $t$ and multiplying by $-1$ one concludes the proof.
\end{proof}

\section{Numerical results} 
\label{sec_num}
In this section   we present numerical results obtained using the  model and its  analysis developed in Section~\ref{sec_net}, and assuming    parameter values  typical for a rural LoRa network.

We first consider  a homogeneous deployment of  nodes 
paramet\-rized by the  mean    number $N_\mathrm{nodes}$ of nodes within a radius of $R=8000\mathrm{m}$ around the base station, which  corresponds to the spatial density  $\lambda_s= \frac{N_\mathrm{nodes}}{\pi R^2}$ of nodes per  per~m$^2$.
Suppose that each node transmits a packet to the base station  every $16.666$ minutes on average, which is equivalent to the rate  $\lambda_t= 0.001$ transmissions per second.
In order to represent this network  we  use the Poisson rain model with space-time intensity  $\lambda:=\lambda_s\lambda_t$ of transmissions per second per m$^2$.
\begin{remark}
Note that when use the Poisson rain  model  we are making an approximation
of the real network. In fact, "real" nodes are fixed in space and send  (assume independent Poisson) streams of packets in time. Nodes in Poisson rain model are not fixed in space: we may see them rather as  ``born'' at some time, transmitting a packet and ``dying'' immediately after.
While it is possible to come up with a model representing fixed nodes in  space, see e.g. the Poisson-renewal model developed in~\cite{blaszczyszyn2010stochastic,blaszczyszyn2015interference}
for non-slotted Aloha, its analysis is more complicated,
with less closed-form results. Moreover, most importantly,  the two
models  provide very close results. This can be theoretically explained by the convergence of the space-time pattern of transmissions in the real situation (with fixed nodes) to the Poisson rain model  when  $\lambda_s\to\infty$  and $\lambda_t\to0$ with $\lambda=\lambda_s\lambda_t=\mathrm{constant}$. 
Numerical results provided in~\cite{blaszczyszyn2010stochastic,blaszczyszyn2015interference}
confirm this statement.
\end{remark}

We use the Hata model~\cite{ha80} to determine a suitable value for the path-loss exponent~$\beta$. In this model the path-loss in dB at  distance $r$  is $ (44.9 -6.55 \hskip 0.1 cm \log_{10} (h_B))  \log_{10}(r)$ where $h_B$ is the 
height of the base station antenna in meters.  We assume that the height of the base station antenna is 
$30$~m and obtain the path-loss function  $35.22 \hskip 0.1 cm \log_{10}(r)$ in dB. Thus we choose~$\beta =3.5$. 
We assume the path loss constant to be $\kappa=0.5$.
Our default assumption regarding propagation effects  is  Rayleigh fading of mean 1. With this choice we have $E(F^{2/\beta})=\frac{2 \Gamma(2/  \beta)}{\beta}$. However, in Section~\ref{subsec.Fading}
we shall also consider the no-fading case and log-normal shadowing.

We assume that all nodes transmit with power $P_\mathrm{tr}=10$~dBm.
Other LoRa parameters are specified in Table~\ref{tab3}.
\begin{table}[ht]
\begin{center}
\begin{tabular}{|c|c|}
 \hline
  BW   &  $125$~kHz   \\
 \hline  
  $\SF$ & $\{6,\ldots,12\}$ \\
   \hline
  $n_p$    &  6 \\
  \hline
 PL  & 20 Bytes\\
  \hline
  H &  H=0 (with header)   \\
  \hline 
   DE &   DE=0 (no low data rate optimization) \\
  \hline 
  CR& $4/5$ \\ 
  \hline
\end{tabular}
\vspace{2ex}
\caption{LoRa parameters used in our model}
\label{tab3}
\end{center}
\end{table}

\noindent
Note that we have $N=7$ different values of the spreading factor~$\SF$ (running from 6 to 12).
The packet and preamble transmission times $B_n$ and $\Delta_n$,
$n=1,\ldots,7$, corresponding to these values 
 are calculated using the expressions presented in Section~\ref{subsec_duration}.

\subsection{Reception probabilities} 
\label{results-probability}
Let us assume received power thresholds (sensitivities) for different spreading factors as
in Table~\ref{tab1}.
That is, $P_1,\ldots,P_7$ expressed in dBm are as in the right column of this table from bottom  to top ($P_1=10^{-13.7}\mathrm{dBm}$,
$P_2=10^{-13.5}\mathrm{dBm}$, etc,  $P_7=10^{-12.1}\mathrm{dBm}$
and the corresponding values of the spreading factor are
$\SF_1=12$,  $\SF_2=12$, etc, $S_7=6$).
Using  the expressions given in Proposition~\ref{prop.Pin} we calculate the probability of successful reception $\Pi_n$ of packets receiving with powers in successive intervals $[P_n,P_{n+1})$, $n=1,\ldots,7$.
These probabilities, calculated assuming  Rayleigh fading, are shown in Figure~\ref{Plot1}. 
Note that increasing $\SF$ values correspond to nodes
received with  smaller powers, thus located statistically in larger distances to the base station.
(When some fading is assumed then this relation is not a  deterministic function of the distance.) 
We observe that, up to the density of $N_\mathrm{nodes}=2000$ nodes 
in the disk of radius $8$~km, all transmissions have a very high reception probability (close to~1), except for the three weakest power categories  ($\SF=10,11,12$).

\begin{figure}[t!]
\begin{center}
\centerline{\includegraphics[scale=0.75]{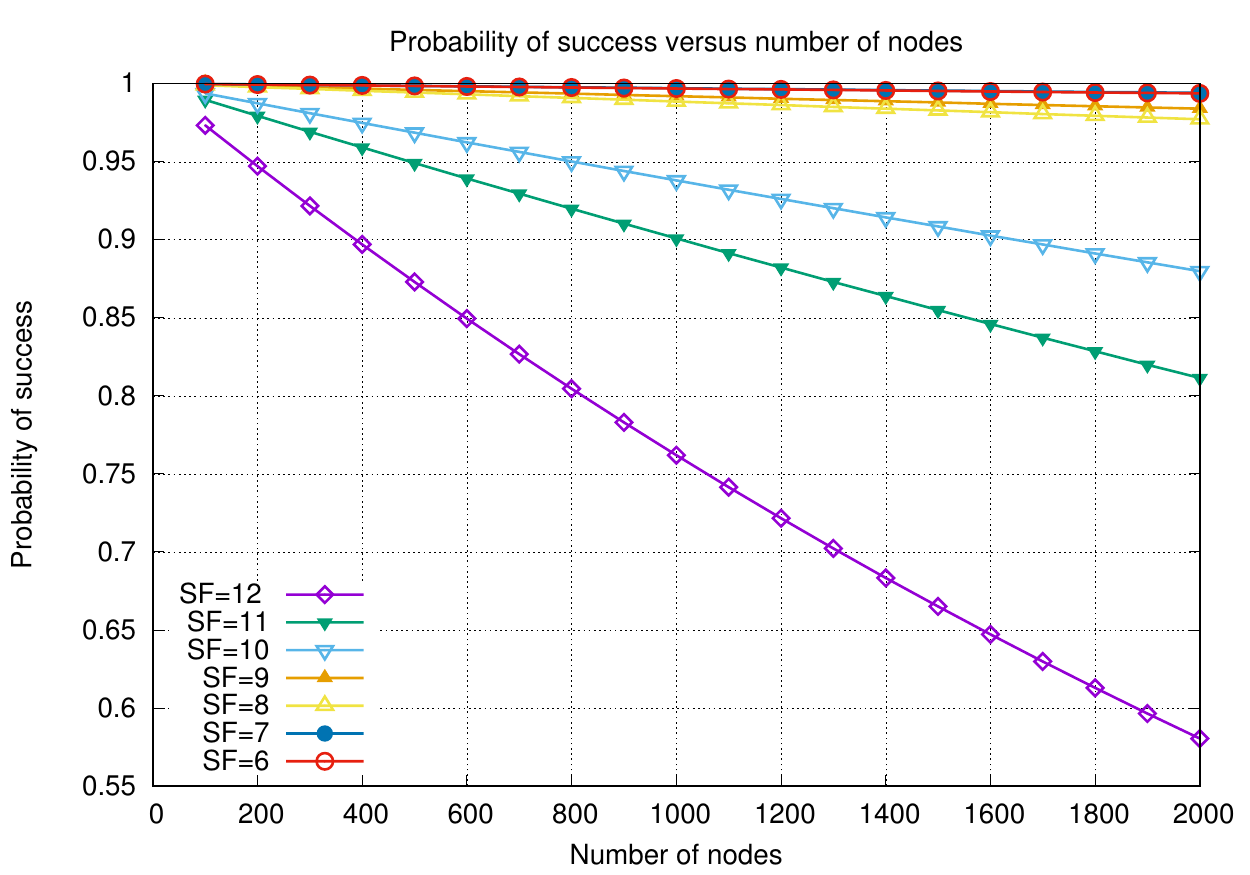}}
\caption{Probability of successful reception versus the node density  for all  spreading factors $\SF$ in the Rayleigh fading case.
\label{Plot1}}
\end{center}
\end{figure}

\subsection{Equalizing reception probabilities for all nodes}
We assume the density of nodes corresponding to $N_\mathrm{nodes}=1000$.
Using the results of 
Section~\ref{subsec_fairness}, 
we calculate  modified values of the sensitivity which allow all nodes to have the same reception probability equal to $\Pi\approx0.95$.
These modified sensitivities  are presented in Table~\ref{tab4}, where we also recall  the values recommended by the provider of LoRa technology.
Observe that for $\SF=7, 8, 9$ we obtain the values smaller than these recommended by  
LoRa provider, whereas for $\SF=10,11,12$ we obtain bigger values. Note that the reception probabilities for these two groups in case of  the recommended configuration are, respectively, larger and smaller than the target value~$\Pi=0.95$. 
Note also that the smallest modified sensitivity is $-135~\mathrm{dBm}>-137\mathrm{dBm}$, which means that 
we will not be able to receive packets with  some small powers
still acceptable for the original configuration,
unless we decrease the target reception probability $\Pi$.
\begin{table}[ht]
\begin{center}
\begin{tabular}{|c|c|c|}
  \hline
           & \multicolumn{2}{|c|}{Sensitivity in dBm}   \\
  SF     & equalizing reception & LoRa recommended   \\
  \hline
  6 &  -121 & -121   \\
   \hline
  7 &  -124 &  -126 \\
   \hline
  8  & -127&  -129 \\
   \hline
  9 & -130  &   -131\\
   \hline
  10 & -133 & -133  \\
   \hline
  11 &  -134& -135\\
   \hline
  12 &  -135 & -137\\
  \hline 
\end{tabular}
\vspace{2ex}
\caption{Sensitivities equalizing 
reception probabilities to $\Pi\approx0.95$ for all spreading factors in the network of density  $N_\mathrm{nodes}=1000$ nodes in the radius of $8$~km, compared to  the sensitivities recommended by LoRa technology provider. Note  that $-135>-137$, and thus packets with some small powers are lost with respect to the  original  configuration.}
\label{tab4}
\end{center}
\end{table}

\subsection{Effect of the fading law} 
\label{subsec.Fading}
Now, we study  the influence of the  propagation effects modeled by 
the distribution of~$F$. We consider three cases: 
no fading i.e. $F \equiv 1$, Rayleigh fading of mean 1 
(exponential $F$) and log-normal fading of mean 1 
($ F = \exp(-\sigma^2/2+\sigma Z )$  with   $Z$  standard  normal).
In our  model this distribution impacts the results only via its moment $\E(F^{2/\beta})$. Recall that
$$
\E(F^{2/\beta})=
\begin{cases}
1 \quad&\text{when no fading,}\\
2 \Gamma(\frac{2}{\beta})/\beta
&\text{for Rayleigh fading,}\\
\exp(\sigma^2(2-\beta)/\beta^2  
&\text{for log-normal shadowing.}
\end{cases}
$$
\noindent
  In Figure~\ref{Plot3} we show how these three different cases of propagation effects impact  the probability of successful transmission of the weakest transmissions,
  that is using the spreading factor  $\SF=12$. 
  We observe that the performance of this category of transmissions  is slightly improved by random propagation effects.  This  can be generalized to all transmissions using Jensen's inequality
$\E(F^{2/\beta}) \le (\E(F))^{2/\beta}=1$   (recall $\beta>2$, $\E(F)=1$)
and observing that the network with some random propagation effects is equivalent to a
sparser network without propagation effects, cf~\cite{netequivalence} and hence generates smaller interference,.

\begin{figure}[t!]
\begin{center}
\centerline{\includegraphics[scale=0.75]{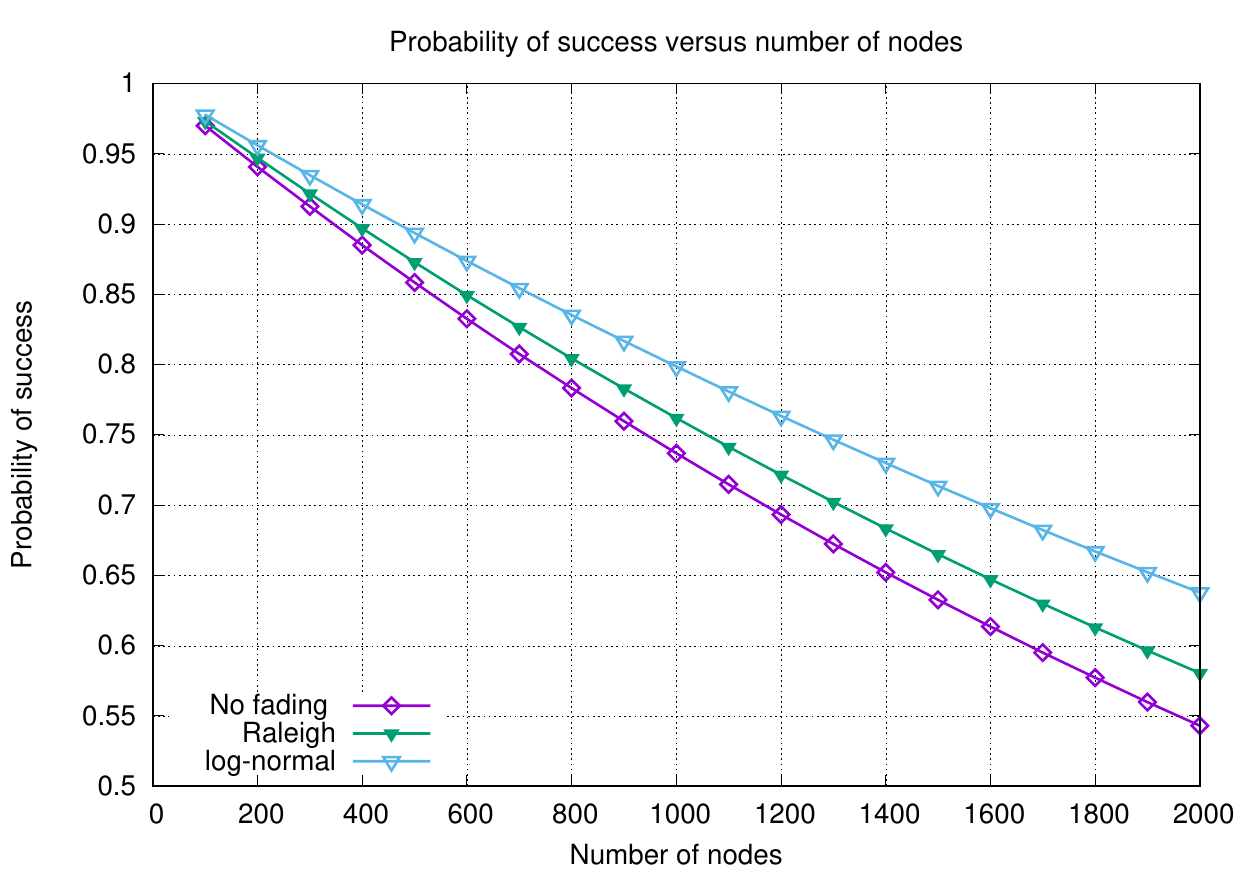}}
\caption{Impact of the propagation effects on the probability of successful transmission of the weakest transmissions ($\SF=12$). Log-normal variance is taken $\sigma=2$dB.
\label{Plot3}}
\end{center}
\end{figure}

\subsection{Effect of decaying density of nodes} 
Finally we consider a  non-homogeneous distribution of nodes,
which we assume  decreasing with the distance $r$ to the base station  as the function $\lambda_sr^{-1/5}$.
Here $\lambda_s$ and all LoRa parametres are assumed as in Section~~\ref{results-probability}.
We use the results of Section~\ref{sub_non_hom} to compute the probability of 
 successful transmission for  various spreading factors. 
The results are shown in Figure~\ref{Plot4}. We observe a very significant improvement and equalization of the probabilities of success for all  categories of transmissions,
with all values presented being bigger than $0.9$.

\begin{figure}[t!]
\begin{center}
\centerline{\includegraphics[scale=0.75]{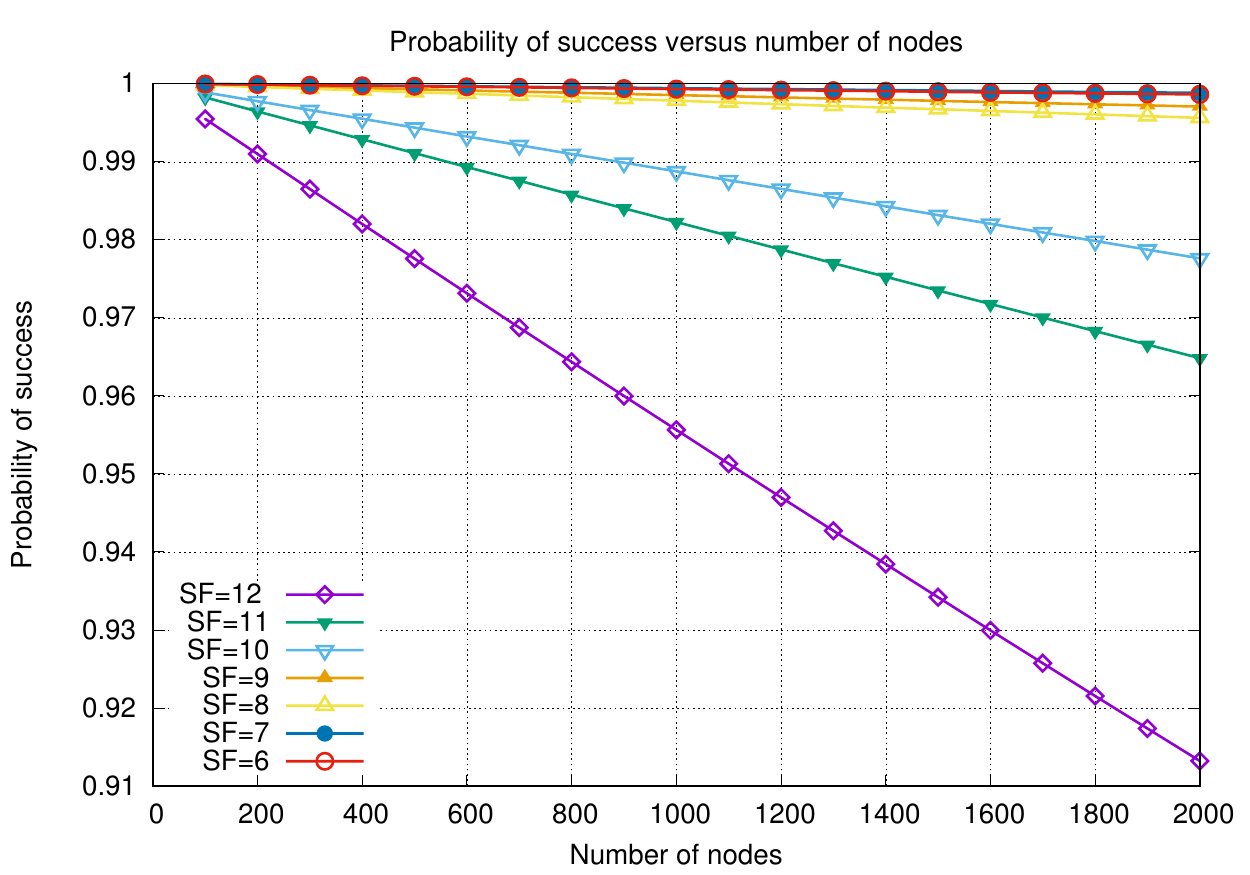}}
\caption{Probabilities of successful transmission
in an  in-homogeneous network with spatial density of nodes 
decaying with the distance to the base station as  $r^{-1/5}$.
\label{Plot4}}
\end{center}
\end{figure}

\section{Comparison to Aloha}
\label{s.Aloha}

Performance of LoRa can be compared to that of the classical, non-slotted Aloha. This can be done assuming 
the same Poisson rain model of packet transmissions.
Let us briefly recall  Aloha model assumptions and results obtained 
in~\cite{blaszczyszyn2008using} (the  context being  that of  sensors sending packets to a sink-node).

\subsection{Aloha SINR collision model}
In Aloha system the base station needs to be synchronized  
before receiving a packet. This can be done when 
 the base station is  not locked on a different packet at the given packet arrival time.
If this condition is satisfied,
the base station starts and  continues receiving the packet until the end of the
transmission. However,  the transmission may be  lost because of the
interference created by other packet emissions started during the reception,
in which case the error will be
detected only at the end of the reception. (These
interfering packets will be lost as well since the base station was
not idle to lock on them at their  arrival epochs.)

\subsection{Packet admission policy}
In order to improve Aloha efficiency, a packet admission policy 
is introduced. 
Once the base station detects a packet transmission
it may decide to receive or ignore the detected packet
according to some (possibly randomized) packet  admission policy,
based, for example, on  the  received power. The idea is to 
allow
the base station to  ignore some very weak packets,
whose successful reception is hardly likely,
as well as to drop some fraction of very strong packets (which will be transmitted at another time) in order to let the base station be  more often available for packets with a moderate received power. 
In the model proposed and studied in~\cite{blaszczyszyn2008using}
this admission policy took a form  of the probability $d(x)$
of  packet drop (even if the base station is available) applied for packets sent from the location $x$ on the plane.
For a given path-loss model it can be related to the mean received power. 
It was shown how to optimize it so as to  maximize the coverage or total
throughput in the network.

\subsection{Aloha model analysis}
Assuming the same Poisson rain model of packet transmission
as used in the present paper, with Rayleigh fading, the above Aloha SINR collision model, 
and some arbitrary probabilistic packet admission policy, it is possible to express the probability of successful reception 
of packets versus the transmission location $x$ on the plane.
The analysis is however more complicated than that of LoRa due to 
a more complicated (detailed) SINR collision model.

More specifically, denoting by  $\lambda$ the space-time density of transmissions, 
$B$  packet transmission time (the same for all transmissions),
the probability of successful reception of a packet transmitted from location $x$, by the base station located at the origin, can be factorized as follows
$$\Pi_x=\frac{d(x)}{1+\lambda B}\,\Pi'_x,$$
where $\frac{1}{1+\lambda B}$ is the (well known Erlang's formula for the)  probability that the packet finds the
base station  idle (ready to lock on the transmission) when it arrives,
$d(x)$ is the probability that it is not dropped by the admission policy, and 
 $\Pi'_x$ is the conditional probability that it is 
correctly  received (SINR condition satisfied), given the base station starts receiving~it.
Explicit evaluation of $\Pi'_x$ requires a fine analysis of the interference (averaged over the reception period) from  different types of packets: 
non-admissible (too weak power to lock on) and admissible ones arriving 
before and during the current reception.
Assuming Rayleigh fading, $\Pi'_x$ can be expressed using the Laplace transforms of these components of the interference and of the noise, as shown in~\cite[Fact~3.2 with the Laplace transforms given by~(9.1), (9.2), (9.3)]{blaszczyszyn2008using}.
These formulas and their approximations developed in~\cite{blaszczyszyn2008using} allow one to solve various optimization  problems regarding  $\Pi_x$ with respect to the packet admission policy $d(x)$.

A detailed comparison of the performance of  LoRa and Aloha using the Poisson rain framework is left for future work.

\section{Conclusion}
\label{sec_conc}
We have presented a simple model  to study 
the performance of the node-to-base-station link in a realistic, spatial deployment of a LoRa network, which is a type of 
low-power, wide-area  technology destined for the Internet of Things.
Crucial components of the model are the space-time Poisson process of packet transmissions (called in the literature the Poisson rain model)
and a packet collision model  accounting for LoRa intra-technology interference treatment, which has already been  validated with respect to measurements in laboratory conditions.
Our model allows for an explicit optimization of the packet reception probabilities with respect to numerous LoRa parameters,
taking into account possibly non-homogeneous node deployment on the plane and various types of wireless channel propagation effects. 
In particular,
it sheds light on how the, crucial to LoRa technology,  link budget thresholds (called sensitivities) impact the performance of different categories of transmissions.
We have studied only one base-station scenario; extensions to 
multi-base station networks are left for future work.
A fair comparison of LoRa  to non-slotted Aloha 
within the proposed framework is also possible and desirable. 


\end{document}